\documentclass[11pt]{amsart}
\usepackage{graphicx,color}
\usepackage{bm}
\usepackage{amsmath}
\usepackage{amstext}
\usepackage{amsfonts}
\usepackage{amssymb}
\usepackage{amsthm}
\usepackage{amsrefs}
\usepackage{url}
\usepackage{hyperref}
\usepackage[utf8]{inputenc} 
\newcommand{\bb}[1]{\mathbb{#1}}
\newcommand{\cl}[1]{\mathcal{#1}}
\newcommand{\ff}[1]{\mathfrak{#1}}

\newcommand{\N}{\mathbb{N}}

\newcommand{\R}{\mathbb{R}}
\newcommand{\C}{\mathbb{C}}
\newcommand{\K}{\mathbb{K}}

\newcommand{\rank}{\operatorname{rank}}

\def\>{\rangle}
\def\<{\langle}

\newcommand{\map}[1]{\mathcal{#1}}
\newcommand{\Tr}{\operatorname{Tr}}

\newcommand{\Vect}{\operatorname{Vec}}

\newcommand{\dist}{\operatorname{dist}}
\newcommand{\id}{\operatorname{\textsl{id}}}
\newcommand{\im}{\operatorname{Im}}

\newcommand{\Sym}{{\operatorname{Sym}}}
\newcommand{\Asym}{{\operatorname{Asym}}}
\newcommand{\ran}{\operatorname{ran}}
\newcommand{\spn}{\operatorname{span}}
\newcommand{\CSEP}{\operatorname{CSEP}}
\newcommand{\EB}{\operatorname{EB}}
\newcommand{\Herm}{\operatorname{Herm}}
\newcommand{\SEP}{\operatorname{SEP}}
\newcommand{\PSD}{\operatorname{PSD}}
\newcommand{\IPT}{\ensuremath{\operatorname{IPT}}}
\newcommand{\PPT}{\ensuremath{\operatorname{PPT}}}

\newcommand{\AND}{\text{ and }}

\newcommand{\FORAL}{\:\text{ for all }\:}

\newcommand{\qand}{\quad\text{and}\quad}
\newcommand{\qfor}{\quad\text{for}\quad}
\newcommand{\qforal}{\quad\text{for all}\quad}

\newcommand{\ol}{\overline}
\newcommand{\wt}{\widetilde}
\newcommand{\hS}{{\widehat {\map S} }}
\newenvironment{spmatrix}{\left(\begin{smallmatrix}}{\end{smallmatrix}\right)}

\makeatletter
\def\thm@space@setup{%
  \thm@preskip=6mm plus 1mm minus 2mm
  \thm@postskip=4mm plus 1mm minus 2mm
}
\makeatother

\theoremstyle{plain}
\newtheorem{thm}[equation]{Theorem}
\newtheorem{theo}[equation]{Theorem}
\newtheorem{conj}[equation]{Conjecture}
\newtheorem{cor}[equation]{Corollary}
\newtheorem{prop}[equation]{Proposition}
\newtheorem{lem}[equation]{Lemma}
\newtheorem{lemma}[equation]{Lemma}

\theoremstyle{definition}

\newtheorem{example}[equation]{Example}

\newtheorem{rem}[equation]{Remark}

\numberwithin{equation}{section}

\usepackage{qcircuit}

\begin{document}

\title[Positive maps and  entanglement in real Hilbert spaces]{Positive maps and  entanglement in real Hilbert spaces}

\author[G.~Chiribella]{Giulio Chiribella} 
\address{Department of Computer Science, The University of Hong Kong, Pok Fu Lam Road, Hong Kong}
\email{giulio.chiribella@cs.ox.ac.uk}

\author[K.R.~Davidson]{Kenneth R. Davidson}
\address{Department of Pure Mathematics, University of Waterloo, Waterloo, ON, Canada N2L 3G1}
\email{krdavidson@uwaterloo.ca}

\author[V.I.~Paulsen]{Vern I.~Paulsen}
\address{Institute for Quantum Computing and Department of Pure Mathematics, University of Waterloo,
Waterloo, Waterloo, ON, Canada N2L 3G1}
\email{vpaulsen@uwaterloo.ca}

\author[M.~Rahaman]{Mizanur Rahaman}
\address{Univ Lyon, ENS Lyon, UCBL, CNRS, Inria, LIP, F-69342, Lyon Cedex 07, France}
\email{mizanur.rahaman@ens-lyon.fr}

\begin{abstract}
The theory of positive maps plays a central role in operator algebras and functional analysis, and has countless applications in quantum information science.  The theory was originally developed  for operators acting on complex Hilbert spaces, and less is known about its variant on real Hilbert spaces.   In this article we study positive maps acting on a full matrix algebra over the reals, pointing out a number of fundamental differences with the complex case and discussing their  implications in  quantum information.    

We  provide a necessary and sufficient condition for a real map to admit a positive complexification, and connect the existence of positive maps with non-positive complexification  with the existence of mixed states that are entangled in real Hilbert space quantum mechanics, but separable in the complex version, providing explicit examples both for the maps and for the states. Finally, we  discuss entanglement breaking   and PPT maps, and we show  that a straightforward   real version of the PPT-squared conjecture is false even in dimension 2.  Nevertheless, we show that the original PPT-squared conjecture implies a different conjecture for real maps,  in  which the PPT property  is  replaced by a stronger property of invariance under partial transposition (IPT).   When the IPT property is assumed, we prove an asymptotic version of the conjecture.  
\end{abstract}

\maketitle

\section{Introduction}

Positive and completely positive maps play a central role  in the theory of operator algebras and in the related fields of dilation theory, 
random matrix theory, free probability, semidefinite optimization theory and quantum information \cites{CMW,optimization,NC}.   
In quantum theory, completely positive maps  are ubiquitous as they represent general  evolutions of quantum states, 
while positive maps are a powerful tool for  characterizing quantum entanglement  \cites{peres1996separability,Horodeckis,horodecki1997separability,terhal2000schmidt,terhal2000bell,sanpera2001schmidt,clarisse2005characterization,ranade2007jamiolkowski,chruscinski2009spectral,skowronek2009cones} (see \cite{horodecki2009quantum} for a review).

Much of the  theory of positive maps  was originally developed in the case where  the underlying number field is the complex numbers.  
In contrast, the case of  operator algebras over the reals  has not been looked at as extensively. 
Notable exceptions are the works of Ruan \cites{Ruan,Ruan1}, and, very recently,  Blecher and  Tepsan  \cite{BlecherTepsan}. 
The real variant quantum theory  was considered in an early work by St\"uckelberg  \cite{stueckelberg1960quantum}, and was later found to differ from standard complex quantum theory in a number of essential features.    
Araki  \cite{araki1980characterization} observed that, while the state space of composite systems  in complex quantum theory  is equal to  the product of the dimension of  the components, in real quantum theory it is strictly larger.   
Dually, this fact implies that product observables are not sufficient to identify general entangled states \cites{wootters1990local,hardy2012limited}.   
Explicitly, Wootters \cite{wootters1990local} provided an example of two mixed states that have orthogonal support, and yet give rise to the same correlations for all possible local observables.  
Building on this example, it was later shown that there exist  evolutions in real quantum theory  that are indistinguishable when acting on a single system, but can be  distinguished perfectly when applied locally on  an entangled state \cite{chiribella2010}.  
These  phenomena  can be summarized by the statement: \textit{real quantum theory violates the Local Tomography Axiom} (see \cite{chiribella2021process} for a review), often assumed in axiomatic reconstructions of standard finite-dimensional quantum theory (see e.g. \cites{Hardy01,CDP11,DB11,MM11,Hardy11}).       
Further differences between real and complex quantum theory were observed in \cites{wootters2012entanglement,aleksandrova2013real,chiribella2013quantum,wootters2014rebit}.    
Very recently, an experimental test that distinguishes between complex quantum theory and real quantum theory was proposed \cite{Renou}. 

In this paper  we develop  the theory of positive maps acting on a full matrix algebra  over the reals, and discuss their applications to the study of entanglement in real quantum mechanics. Our motivation stems from a fact that was observed recently in \cite{CDPR}: complexification of a positive isometry on a real space is not necessarily an isometry in the complex case. This leads to a natural question: whether complexifications of real positive maps are always positive over the complex numbers?   
To answer the question we characterize the subset of  real maps with  a positive complexification, providing explicit examples of positive real maps outside this subset.  
Then, we prove that the existence of positive real maps with non-positive complexification is equivalent to the existence of quantum states that are entangled in real quantum theory, but are separable when regarded as elements of the larger state space of  complex quantum theory. 
These states, of which we also provide explicit examples, cannot violate Bell inequalities  \cite{bell1964einstein}  (see \cite{brunner2014bell} for a review), 
even if  pre-processing operations are allowed \cite{popescu1995bell} and if an arbitrarily large number of identically prepared systems is available \cites{navascues2011activation,palazuelos2012superactivation}.         
Operationally,  the existence of real-entangled, but complex-separable states  can be interpreted in a resource-theoretic framework 
\cite{wu2021operational}, as the ability to prepare a larger class of states by performing  a larger set of local operations that are not restricted to have real matrix elements in a given basis.

Finally, we characterize the   maps that break  quantum entanglement on the real domain,  extending a classic characterization by Horodecki, Shor and Ruskai \cite{HSR}, and exploring the relation between entanglement breaking maps and another class of completely positive maps, called PPT, for which the Choi matrix has positive partial transpose.
In the complex case, Christandl  conjectured (see  \cite{RJKHW} ) that $\Phi \circ \Phi$ is entanglement breaking for every PPT map $\Phi$ and in \cite{CMW} (also see \cite{CYT}) it was shown that this conjecture (now known as the PPT-squared conjecture) holds in dimension $3$. 

In stark contrast, here we show that a straightforward real-vector-space analogue of the PPT-squared conjecture is false in even dimension 2: PPT  maps defined on real matrix algebras  do not necessarily square to maps that break entanglement in the real domain.   
Nevertheless, we show that the (complex) PPT-squared  conjecture implies another conjecture, which  we call the IPT-squared conjecture.   The conjecture is that   $\Phi \circ \Phi$ breaks entanglement on the real domain for every completely positive map $\Phi:  M_n (\R) \to M_n (\R)$ whose Choi matrix   is invariant under partial transposition.
We  call such maps IPT and prove that the IPT-squared conjecture holds asymptotically. 
Moreover, when the IPT maps under consideration are trace preserving and unital, we prove that a finite number of iterations yields a map that is  entanglement breaking on the real domain.

The rest of the paper is structured as follows.  
In Section \ref{sec:defpos} we discuss the notion of positivity on real matrix algebras, pointing out the need to include 
 commutation with the adjoint
 in the definition of a positive map.  
In Section \ref{sec:choi},  we  consider  $p$-positive maps and characterize their Choi matrices. 
The results of this section are mostly adaptations of existing results in the complex case. 
In Section \ref{sec:poscomplex} we characterize the set of $p$-positive real maps that admit a $p$-positive complexification, 
and in Section \ref{sec:examples1} we provide concrete examples of $p$-positive real maps that violate this condition.   
We then move to the study of entanglement on real Hilbert spaces, characterizing entanglement witnesses and showing that the existence of positive real maps with non-positive complexification is equivalent to the existence of mixed states that are entangled in real quantum theory but separable in the complex version (Section \ref{sec:witnesses}).  
Examples of real-entangled but complex-separable quantum states are provided in Section \ref{sec:examples2}.  
In Section \ref{sec:breaking} we consider the notion of entanglement $p$-breaking channels, which transform arbitrary entangled states into states with Schmidt rank at most $p$, and show an explicit example of a map that breaks complex-entanglement but preserves real-entanglement.      
Finally, in Section \ref{sec:PPT}, we conclude with a discussion of the relation between entanglement breaking channels and PPT channels. We show that in the real case there are PPT channels whose square is not real entanglement breaking. We then propose an alternative to the PPT-squared conjecture in the real case, that we call the IPT-squared conjecture. We show that several results concerning the PPT-squared conjecture have direct analogues for the IPT-squared conjecture that are also true.

\section{Positivity preservation and commutation with the adjoint}\label{sec:defpos} 

Let $\bb K$ be either $\C$ or $\R$. 
A matrix    $P \in M_n(\K)$ is called positive (denoted $P\ge0$) if it is of the form $P  =  A^*  A$ for some $A\in  M_n (\K)$, where $A^*$ is the adjoint of $A$ 
(for $\bb K =  \R$, the adjoint  $A^*$ coincides with the transpose $A^t$).    
 The set of all positive matrices in $M_n(\bb K)$ will be denoted by $\PSD(\bb K^n)$.

A matrix is Hermitian if $A=A^*$, and we write $\Herm(\bb K^n)$ for the set of Hermitian matrices in $M_n(\bb K)$.  
For a Hermitian matrix $A \in \Herm(\bb K^n)$, positivity is equivalent to the condition $\<  v  | A  v\>  \ge0$ for every $v\in  \K^n$.    
Here, the product is defined as $  \<  v  |  w\>  =  v^*  w$, regarding the vectors $v,w \in  \K^n$ as column vectors.

We say that a linear map $\Phi :  M_n  (\K)  \to  M_{m}  (\K)$ is {\em positivity preserving} if $\Phi  (P)  \ge0$ for all $P\ge 0$, 
and {\em commutes with the adjoint} if $\Phi  (A^*)   =  \Phi(A)^*$ for all $A\in M_n (\K)$.    In the complex field, commutation with the adjoint is equivalent to the condition that $\Phi$ maps Hermitian matrices into Hermitian matrices. For this reason,   maps that commute with the adjoint are sometimes called {\em Hermitian preserving}. In the real field, however, commutation with the adjoint is a stronger property than preservation of Hermiticity:   a map $\Phi  $  commutes with the adjoint  if and only if 
\[
\Phi \left( \Sym_n(\R)  \right)  \subseteq   \Sym_m(\R)  
\qand  \Phi   \left( \Asym_n(\R)  \right)   \subseteq  \Asym_m(\R)  \,, 
\]
where  $\Sym_n(\R) = \{ A \in M_n(\R) : A^t = A \}$  and   $\Asym_n(\R) = \{ A \in M_n(\R) : A^t = -A \}$ are the spaces of symmetric and antisymmetric matrices, respectively.   Later in the paper (Lemma~\ref{lemma-pos}), we will show  that $\Phi$ commutes with the adjoint if and only if its Choi matrix is Hermitian. 

If the map $\Phi$ is positivity preserving and  commutes with the adjoint, we call it  {\em positive}.    
In the complex field,   the set of positive maps coincides with the set of positivity preserving maps, because preservation of positivity  implies  commutation with the adjoint.  This implication, however, does not hold on the real field.  
Counterexamples are abundant. For example, the map  
\[
 \Phi   :     M_2 (\R)  \to  M_1  (\R),   \,\, 
 \begin{pmatrix} a  &  b  \\  c  &  d \end{pmatrix}   \to   (a+ d  + b-c )/2 
\] 
is positivity preserving but  does not commute with the adjoint.

In special cases, commutation with the adjoint follows from positivity preservation, plus additional properties, such as unitality and unit norm:

\begin{prop}
If $\Phi: M_n(\R) \to M_{m}(\R)$ is positivity preserving, unital and (operator) norm 1, then $\Phi$  is commutes with the adjoint.
\end{prop}

\begin{proof}
Since $\Phi$ is positivity preserving,   $\Phi    \left( \Sym_n(\R)  \right)  \subseteq   \Sym_m(\R)$.     
We now show that unitality and the norm-1 condition imply that   $\Phi   \left( \Asym_n(\R)  \right)  \subseteq  \Asym_m(\R) $. 

The proof is by contradiction. Let $C$ be an antisymmetric matrix, with $\|C\|  =1$ without loss of generality.  
Suppose that $\Phi(C) = X$ is not antisymmetric. Then $B = \frac12(X+X^t) \ne 0$.
Since $B^t=B \ne0$, the spectrum of $B$ is real and not $\{0\}$.
Let $\lambda$ be an eigenvalue of $B$ with $|\lambda| = \|B\|$.

Since $C$ is antisymmetric and norm 1, we have  $-I_n \le C^2 \le 0$ and $-1$ is an eigenvalue of $C^2$.
For $t\in \R$, 
\[
 \| I_n + tC \|^2 = \| I_n - t^2C^2\| = 1 + t^2.
\]

Now choose the sign of $t$ so that $|t\lambda|$ belongs to the spectrum of $tB$; then 
\[
 \| \Phi(I_n + tC) \| = \| I_m + tX \| \ge   1 + |t\lambda| . 
\]
For sufficiently small $t$, this yields a contradiction with the norm 1 property of $\Phi$. 
\end{proof}

One may ask if commutation with the adjoint follows from positivity preservation and from the property $\|  \Phi \|  =  \| \Phi  (I_n) \|$, without requiring unitality. 
The following counterexample answers the question in the negative.

\begin{example}
We provide an example of a positivity preserving  map  $\Phi$   on $M_2 (\R)$ that satisfies the condition  $\|  \Phi \|  =  \| \Phi  (I_2) \|$ but  does not commute with the adjoint.  
Define
\[
 \Phi (A)  =  \Phi  \begin{pmatrix}  a& b  \\  c   & d\end{pmatrix} =    \begin{pmatrix}   \frac{a+d}2 &  0 \\  0  & \frac{b-c}2 \end{pmatrix}\,.
\]   
 The map $\Phi$  is positivity preserving, since $\Phi  \begin{spmatrix} a &  b \\  b & d\end{spmatrix}  =  \frac  {a+d}2 \,  E_{11}$, for any positive matrix.  Moreover, it satisfies 
\begin{align*}  
 \left\|  \Phi  (A)  \right\| =  \max \left\{  \left| \frac{a+d}2\right|\, ,   \left|  \frac{b-c}2\right|\right\}  
 \le  \left\|  A \right\| 
\end{align*}
Hence, $\|  \Phi \|  = \|  \Phi  (I_2)\|=1$. However, $\Phi$ maps antisymmetric matrices into symmetric ones, and therefore does not commute with the adjoint. 
\end{example}

We  conclude with an elementary lemma that will be useful later.

\begin{lemma}\label{lem:rankoneisenough}
Let  $\Phi  :  M_n ({\K})  \to  M_m ({\K})$  be a map that commutes with the adjoint.  Then, $\Phi$ is positive if and only if   $  \< w   |  \Phi  (  v v^*)   w\>\ge 0$ for every pair of vectors $  v\in  \K^n$ and $w\in\K^m$, where $vv^*$  is the rank-1 matrix with entries $  (vv^*)_{ij} :  =   v_i  \overline v_j$,  $\overline v_j$ denoting the complex conjugate of $v_j$.   
\end{lemma}

\begin{rem}
Since the span of the cone $\PSD(\R^n)$ is only $\Sym_n(\R)$, some authors prefer to study maps on this smaller domain in the real case.  However, to better contrast the real and complex case, we find it more convenient to keep the domain $M_n(\K)$,   and to allow $\K$ to be either $\R$ or $\C$.  Note that maps defined on $M_n (\R)$ have a uniquely defined complexification, which allows us to compare the properties of a real map with the properties of its complex version.  In contrast, the complexification is not uniquely defined for maps defined only on $\Sym_n(\R)$.  For example,  the maps $X \mapsto X$ and $X \mapsto X^T$ are identical    as maps on $\Sym_n(\R)$, but the former is completely positive while the latter is not.   For applications in quantum information, it is also worth noting that $\Sym_n(\R) \otimes \Sym_m(\R) \subsetneq \Sym_{mn}(\R)$ as subspaces of $M_n(\R) \otimes M_m(\R) \simeq M_{mn}(\R)$, so the whole theory of  maps on composite systems becomes fractious if one attempts this viewpoint. 

Similarly, some authors \cite{choi:biquad}, \cite{KMcCSZ} prefer to study {\it biquadratic polynomials}. Given a linear map $\Phi: M_n(\R) \to M_m(\R)$ the biquadratic polynomial that it generates is the polynomial in $n+m$ variables given by
\[ p_{\Phi}(x,y) = y^t \Phi(xx^t) y, \,\, x \in \R^n, \, \, y \in \R^m.\]
But $p_{\Phi} = p_{\Psi} \iff (\Phi - \Psi)(\Sym_n) = (0),$ so these polynomials are in a one-to-one correspondence with maps from $\Sym_n(\R)$ into $\Sym_m(\R)$ and again do not capture the extension of such maps to the full matrix algebras. 

We would also like to point out that positive maps on $\Sym_n(\R)$ to itself behave ``nicely" in the sense that most of the well-known results for positive maps on complex matrices continue to hold in this setting. For example, for the complex field, a result of  Russo and Dye  (Corollary 2.9 in \cite{paulsenbook}) states that a positive map on a unital C*-algebra attains its norm at the identity. This result is false in the real case (i.e maps on $M_n(\R)$)as the positive matrices don't span the vector space $M_n(\R)$. However, this result is true for positive maps defined on $\Sym_n(\R)$. The proof is a minor modification of the original one and we choose not to include here. 

One shortcoming of our approach is that maps like $\Phi(X) = X - X^t$, have the property that both $\Phi$ and $- \Phi$ are positive. 
Thus, the cone of positive maps on $M_n(\R)$ is not pointed.
However, as we show in Remark~\ref {R:pointed cones}, the cone of $p$-positive maps is pointed for $p\ge2$.
\end{rem}

\section{$p$-positive maps and their Choi representation}\label{sec:choi} 

A linear map $\Phi  :  M_n ({\K})  \to  M_m ({\K})$ is  {\em $p$-positive} if  the map 
$  \Phi  \otimes  \id_{p}:   M_n ({\K})   \otimes M_p ({\K})  \to  M_m ({\K}) \otimes M_p ({\K}) $ is positive, where $\id  _p$ is the identity map on $M_p ({\K})$.   
Positive maps correspond to the special case $p=1$.   A map is {\em completely positive} if it is $p$-positive for every $p\in  \N$.

 A linear map $\Phi  :  M_n ({\K})  \to  M_m ({\K})$ is conveniently represented in terms of its Choi matrix~\cite{choi1975completely},  defined as  
\begin{align*} 
C_\Phi:=\sum_{i,j} E_{i,j}\otimes \Phi(E_{i,j}) \,,
\end{align*} 
where $E_{i,j}$ are the standard matrix units of $M_n ({\K}) $.     A related representation is the Jamio\l kowski matrix~\cite{jamiolkowski1972linear} 
$J_\Phi:=\sum_{i,j} E_{i,j}\otimes \Phi(E_{j,i})$, originally introduced in the study of positive maps.

We will  often use the following lemma,  the complex version of which is well known.

\begin{lemma}\label{lemma-pos}
A linear map $\Phi:M_n({\K})\rightarrow M_m({\K})$ is 
\begin{enumerate}
\item[(i)]  commuting with the adjoint  if and only if its Choi matrix $C_\Phi$ is Hermitian,
\item[(ii)] $p$-positive if and only if $C_\Phi$ is Hermitian and $ \<  V  |  C_\Phi  V\>\ge 0$ for  every  vector  $V  \in  \K^n\otimes \K^m$ with Schmidt rank at most $p$.  
\item[(iii)] $p$-positive if and only if its adjoint map $\Phi^*$ is $p$-positive.    
\end{enumerate}
\end{lemma}

\begin{proof}
The complex version  of (i) was proven in \cites{de1967linear,ranade2007jamiolkowski},  while the complex version of (ii) was proven in  \cite{jamiolkowski1972linear} 
(for $p=1$, see also \cite{stormer-book}*{Proposition~4.1.11}) and \cites{ranade2007jamiolkowski,chruscinski2009spectral} (general $p$).   
 The proofs of the real versions are sufficiently similar to those in the complex case that we omit them.
    
To prove (iii), one can use   the relation
\begin{align}\label{choiinverse}
 \Tr (C_\Phi (A^t\otimes B)) = \Tr (\Phi(A)B) ,
\end{align}
valid for arbitrary matrices  $A\in M_n({\K})$ and $B\in M_m({\K})$.  
  Using this relation, one obtains  
\begin{align*}
  \Tr  (    C_{\Phi^*}    (  B^t  \otimes A ))    
  &=   \Tr  ( \Phi^*  (B)   A  )   
  =  \Tr  (  B  \Phi  (A))   \\&
  =   \Tr(   C_\Phi   (   A^t \otimes B)) 
  =  \Tr  (     C_\Phi^t  \, S     (  B^t \otimes A) S^*  )\, ,
\end{align*}
where $S : \K^m \otimes \K^n \to \K^n \otimes \K^m $  is the linear operator defined by $S ( w\otimes v)  =  v\otimes w$ for all $v \in  \K^n$ and $w\in \K^m$.  
Since $A$ and $B$ are arbitrary, we conclude 
\begin{align}\label{choiadjoint}
 C_{\Phi^*} = S^* C_\Phi^t S .
\end{align}  
Clearly, $C_{\Phi^*}$ is Hermitian  if and only if $C_\Phi$ is, and  $\<  W  |   C_{\Phi^*}  W\>  \ge 0 $   for every vector $W  \in  \K^m\otimes \K^n $ 
with Schmidt rank $\le p$ if and only if $\<  V  |   C_{\Phi}  V\>  \ge 0 $   for every vector $V  \in  \K^n\otimes \K^m $ with Schmidt rank $\le p$. 
Hence, part (ii) of this lemma  implies that $\Phi^*$ is $p$-positive if and only if $\Phi$ is.  
\end{proof}

\section{Real maps with $p$-positive complexification}\label{sec:poscomplex}

Given a linear map $\Phi: M_n(\R) \to M_m(\R)$ its {\em complexification} is the  linear map
$\wt{\Phi}: M_n(\C) \to M_m(\C),$ defined by
\[
 \wt{\Phi}(X+iY) = \Phi(X) + i \Phi(Y) \qforal X, Y \in M_n(\R).
\]

We now characterize the set of  maps with a $p$-positive complexification.  
The characterization uses  the correspondence between  matrices  $A =  (a_{ij}) \in  M_{n,m}  (\K)$ and 
vectors  in $\Vect  (A)  \in  \K^n\otimes \K^m$ given by  
\begin{align*} 
 \Vect (A)   : =   \sum_{i=1}^n \sum_{j=1}^m  \,  a_{ij}  \,  e^{(n)}_i\otimes e^{(m)}_j ,
\end{align*} 
where, for $r \in  \{m,n\}$,   $\{e^{(r)}_i\}_{i=1}^r$  is the standard orthonormal basis for $\K^r$.

\begin{theo}\label{theo:ChoiCondition}
A map $\Phi  :  M_n (\R)  \to  M_m  (\R)$  has a $p$-positive complexification  if and only if its Choi matrix   $C_\Phi$ is Hermitian and satisfies the condition 
$\<  \Vect  (X)  |  C_\Phi  \Vect (X)  \>    +  \<  \Vect  (Y)  |  C_\Phi  \Vect (Y)  \>  \ge 0$ 
for every pair of matrices $X  ,  Y  \in  M_{n,m} (\R)$ such that $\rank (X+  i  Y)  \le p$.
\end{theo}

\begin{proof}   
By Lemma \ref{lemma-pos},   $\widetilde \Phi$ is $p$-positive if and only if its Choi matrix $  C_{\widetilde \Phi}   =  C_\Phi$ is Hermitian 
and satisfies the condition $\<   V  |  C_\Phi   V\> \ge0 $ for every vector $V  \in  \C^{n}\otimes \C^m$ with Schmidt rank at most $p$.    
In the vector notation, this condition is equivalent to   $\<  \Vect (A)  |  C_\Phi  \Vect  (A)  \>  \ge 0 $ for every matrix 
$A  \in  M_{n,m}  (\C)$ with $\rank  (A) \le p$.  
Writing $A   =  X  +  i  Y$, with $X,  Y  \in  M_{n,m }  (\R)$, and using the fact that  $C_\Phi$ is a symmetric matrix, we obtain   
\[
 \<  \Vect (A)  |  C_\Phi  \Vect  (A)  \>    =  \<  \Vect (X)  |  C_\Phi  \Vect  (X)  \>  +  \<  \Vect (Y)  |  C_\Phi  \Vect  (Y)  \> .  \qedhere
\]
\end{proof}

\begin{cor}\label{cor:2p}
If a  map  $\Phi: M_n(\R) \to M_{m}(\R)$   is $2p$-positive, then its complexification is $p$-positive.  
\end{cor}  

\begin{proof}
Immediate from Theorem \ref{theo:ChoiCondition}, Lemma \ref{lemma-pos}, and the fact that, for $A  =  X+  i Y$,   
$\rank  (X)  =  \rank  (A +  \overline A)   \le  2 \rank (A)$, and similarly for $\rank  (Y)$. 
\end{proof}

\begin{rem} \label {R:pointed cones}
We saw earlier that the cone of positive maps on $M_n(\R)$ is not pointed. 
Now, Corollary \ref{cor:2p}  implies that the cone of $2$-positive maps is pointed, since if $\Phi$ is $2$-positive, 
then $\tilde{\Phi}: M_n(\C) \to M_m(\C)$ is positive, which implies that $-\tilde{\Phi}$ is not positive, and hence $-\Phi$ is not $2$-positive.  
Thus, for every $p>1$ the cone of $p$-positive maps between two real matrix algebras is pointed.    This observation was made by one of the referees of this paper. 
\end{rem}

Note that the Corollary \ref{cor:2p}  establishes that if $\Phi$ is completely positive, then the complexification is also completely positive, a fact proved in Lemma 2.3 in \cite{BlecherTepsan}. 
In particular, this corollary shows that  $2$-positivity on the real field is sufficient for positivity of the complexification.  
In general, however, $2$-positivity is not necessary: for example, the maps   $\Lambda_q  :  M_n (\R) \to M_n(\R)$ defined by 
\begin{align*} 
 \Lambda_q    (A)   =   \Tr(A)\,   I_n     -  q   A    
\end{align*} 
have $p$-positive complexification whenever $q  \le 1/p$ \cite{takasaki1983geometry}.
They are not $(p+1)$-positive (and therefore not $2p$ positive) for $q  > 1/(p+1)$, as one can verify from the relation 
$\<  V_l |  C_{\Lambda_q}  V_l\>  =  1-  q l    $, with  $V_l  :=  \sum_{i=1}^{l}  \,  e_i^{(n)} \otimes e_i^{(n)}/\sqrt {l}$.    
Similar counterexamples can be constructed using the spectral methods developed by Chru{\'s}ci{\'n}ski and Kossakowski in \cite{chruscinski2009spectral}.

\section{Examples of positive real maps with non-positive complexification}\label{sec:examples1}

In this section, we provide several examples of positive maps whose complexification is not.
The first deals with $p$-positivity for general $p$.
The others deal with $p=1$.

\begin{example}
We provide an example of a $(2p\!-\!1)$-positive real map such that its complexification is not $p$-positive. This example shows that Corollary~\ref{cor:2p} is sharp.

Define $\Gamma_q:M_{2p}(\R) \to M_{2p}(\R)$ by 
\begin{align*}
  \Gamma_q(A) =   \Tr ( A)   I_{2p}  -   q (     O_+  A  O_+  +  O_-  A  O_-  ) ,
\end{align*}    
where $O_+   = \begin{spmatrix}   0  &  I_p  \\  I_p  &  0  \end{spmatrix}$ and $O_-   = \begin{spmatrix}    I_p  &  0  \\   0&   -I_p   \end{spmatrix} $. 
The Choi matrix of this map is 
\[
 C_{\Gamma_q}   =  I_{2p}    \otimes  I_{2p}   -   q  \big( \Vect  (O_+)   \Vect  (O_+)^t   +  \Vect  (O_-)   \Vect  (O_-)^t \big) .
\]
Note that the vectors  $\Vect  (O_+)$ and $\Vect  (O_-)$ are orthogonal and each has length $\sqrt{2p}$,
so that 
\[
 \cl P = \frac1{2p} \big( \Vect  (O_+)   \Vect  (O_+)^t   +  \Vect  (O_-)   \Vect  (O_-)^t \big)
\]
is a projection.
Thus $C_{\Gamma_q}   =  I_{2p}    \otimes  I_{2p}   -   2pq \cl P$.
Moreover the unit vectors in $\ran \cl P = \spn\{ O_+, O_- \}$ are $\{ \frac1{\sqrt{2p}} (\cos\theta\, O_+ + \sin\theta\, O_- ) : \theta\in [0, 2\pi] \}$.
Now $O_\theta := \cos\theta\, O_+ + \sin\theta\, O_- \in \cl O_{2n}$ is an orthogonal matrix for all $\theta$.  

Given a unit vector $V = \Vect(A)$, so that $\Tr (A^tA) = 1$, one has 
\begin{align*}
 \<  V   |  C_{\Gamma_q}  V\>  &= \|V\|^2 -  2pq \| \cl P V \|^2 \\&
 = 1  - 2p q  \max_{\substack{W = \cl P W\\ \|W\|=1}}   \<  V |  W\>^2 \\&
 = 1 - \sqrt{2p} q \max_{\theta \in [0,2\pi]} \Tr(A^tO_\theta) .
\end{align*}
If the unit vector $V$ has Schmidt number at most $2p-1$, i.e., $\rank  (A)  \le 2p-1$, we obtain 
\[
  \Tr(A^t O_\theta )  \le  \Tr |A|   \le  \sqrt{ 2p-1} . 
\]  
Hence, we obtain the inequality $\<  V   |  C_{\Gamma_q}  V \> \ge   1-  q  \sqrt{2p(2p-1)}$ for every vector $V$ with Schmidt rank $\le 2p-1$.    
By Lemma \ref{lemma-pos}, we conclude that the map $\Phi$ is $(2p \!-\! 1)$-positive for every $\sqrt{ 2p(2p-1)}q \le 1$.  
  
Take $q$ with $\frac1{2p} < q \le \frac 1{\sqrt{ 2p(2p-1)}}$.
Then $\Gamma_q$ is  $(2p \!-\! 1)$-positive, but its complexification $\widetilde {\Gamma_q}$ is not $p$-positive.
To see this, let 
\[
 V = \frac1{2\sqrt{p}} \Vect(O_+ + i O_-) .
\]
Then $\rank(O_+ + i O_-) = p \rank \begin{spmatrix} i & 1 \\ 1 & -i \end{spmatrix} = p$.
Compute
\[
 \<  V |C_{\widetilde {\Gamma_q}}  |V\>   = \<  V |C_{\Gamma_q}  |V\> =  1 - 2pq < 0 .
\]
Therefore $\widetilde {\Gamma_q}$ is not $p$-positive by  Lemma \ref{lemma-pos} again.
\end{example}

In the rest of the section we focus on the $p=1$ case.

\begin{example}
Let us start with a simple example of positivity preserving real map with non-positive complexification.  
Define $\Phi_s:M_n(\R) \to M_n(\R)$ by
\begin{align*} 
 \Phi_s(A) = A + s(A-A^t) .
\end{align*}    
Then, $\Phi_s$ is a positivity preserving for all $s \in \R$. 
However it  does not commute with the adjoint if $s \ne 0$. 
Hence $\wt{\Phi}_s$  does not commute with the adjoint, and thus is not positive.
Explicitly, if we set 
\[
 P = E_{1,1}+E_{2,2} + i (E_{1,2} -  E_{2,1}) \ge 0 , 
\]
where $E_{i,j}$ are the standard matrix units, then 
\[
 \wt{\Phi}_s(P) = E_{1,1} + E_{2,2} + (1+2s)i(E_{1,2}-E_{2,1}) ,
\]
which  has eigenvalue $-2s$, and therefore is not positive for $s>0$. 
\end{example}

For the complex field, a result of   Russo and Dye  (Corollary 2.9 in \cite{paulsenbook})  states that a positive map 
$\Phi$ on a unital C*-algebra satisfies the condition $\|\Phi\| = \| \Phi(I) \|$.  
Hence, a necessary condition for positivity of the complexification  is that  $\Phi$  is positive  and satisfies the Russo-Dye condition $\|  \Phi  \|   =  \|  \Phi  (I) \|$.    
We now provide  several counterexamples, of increasing difficulty, showing that the Russo-Dye condition is not sufficient for positive complexification.

\begin{example}
Here we provide an example of a  positive, unital, norm one map on $M_2(\R)$ whose complexification is not positive.
Define
\[
 \Phi(A) = \Phi \begin{pmatrix} a & b \\ c & d\end{pmatrix} =
 \begin{pmatrix} d & \frac{b-c}2 \\ \frac{c-b}2 & d\end{pmatrix} .
\]
This map is clearly unital and commuting with the adjoint.
It is positive preserving since $\Phi \begin{spmatrix} a & b \\ b & d\end{spmatrix} = d I_2$.
Calculation shows that $\|  \Phi  \|   =  1 = \|\Phi(I)\|$.  
Finally, the non-positivity of $\wt \Phi$ can be seen from the relation 
$  \wt \Phi  \begin{spmatrix}         \lambda    &     i   \\   -i    &  1/\lambda   \end{spmatrix}   
=   \begin{spmatrix}         1/\lambda    &     i   \\   -i    &  1/\lambda   \end{spmatrix}$ for   $\lambda  >  1$. 
\end{example}

\begin{example}
Here we provide an example of a positive trace-preserving  map on $M_2(\R)$ that satisfies the Russo-Dye condition but has a non-positive complexification. 
Define
\begin{align*} 
 \Psi  (A)   &=  \Psi \begin{pmatrix}      a  &   b   \\   c    &  d      \end{pmatrix}  
 :=    \begin{pmatrix}      0    &    \frac{b-c}2   \\   \frac{c-b} 2   &  a+d      \end{pmatrix} \\&
\end{align*} 
One has  $\|  \Psi  \|    =  \|  \Psi  (I_2) \|  =  2$.    
On the other hand,  $\wt \Psi$  is not  positive, as one can see from the relation 
\[
  \wt \Psi   \begin{pmatrix}      1    &    i   \\   -i    &  1      \end{pmatrix}    =    \begin{pmatrix}      0    &    i   \\   -i    &  2      \end{pmatrix}    .
\]
\end{example}

\begin{example}
Finally, we provide an example of  positive, trace-preserving and unital map on $M_3 (\R)$ that  
satisfies the Russo-Dye condition, but has a non-positive complexification.  
Define the one-parameter family of maps for $t\in [0,1]$, 
\begin{align*} 
 \rho_t (A)   &
 = \rho_t   \begin{pmatrix}      a_{11}  &   a_{12}      & a_{13}    \\    a_{21}  &   a_{22}    & a_{23}   \\     a_{31}  &   a_{32}   &  a_{33}     \end{pmatrix}  
 :=   \begin{pmatrix}      
  \frac{  a_{11 }+  a_{22}  +  2  a_{33}  }4  &  \frac{  (a_{12}  -  a_{21})  t}{2 \sqrt 2}    &  0  \\   
  \frac{  (a_{21}  -  a_{12})  t} {2 \sqrt 2}    &    \frac{  a_{11 }+  a_{22}  +  2  a_{33}  }4   &  0      \\  
  0  &  0  &   \frac{  a_{11} +  a_{22} }2     
 \end{pmatrix} .
\end{align*} 
This map is  clearly unital, trace-preserving, and positive.

Set $x  = \frac{a_{11}  +  a_{22}}2$,   $y  = \frac{a_{12}  -  a_{21}}2$,  $z =  a_{33}$,  and  $r  =  \max  \{   \sqrt{x^2  +  y^2}, |z| \}$.
Compute
\begin{align*}
 \|   A  \|   &\ge  \left \| \begin{pmatrix} a_{11}  &  a_{12} & 0   \\  a_{21}  &  a_{22} & 0  \\  0  &  0  &  a_{33}  \end{pmatrix} \right\|  \\&
 \ge \max \Big\{ \left\| \frac12 \begin{pmatrix}  a_{11}  & a_{12} \\  a_{21}  &   a_{22}  \end{pmatrix}
 + \frac12 \begin{pmatrix}  a_{22}  & -a_{21} \\  -a_{12}  &   a_{11}  \end{pmatrix} \right\|, \ |z| \Big\} \\&
 = \max \Big\{ \left\| \begin{pmatrix}  x  & y \\  -y  &   x \end{pmatrix} \right\|, \ |z| \Big\}
=  \max  \left\{  \sqrt{ x^2   +  y^2 }  ,  |z| \right\}    = r   .
\end{align*} 
The three columns of $\rho_t(A)$ are orthogonal. Hence
\begin{align*}
 \|  \rho_t  (A)\|  &= \left \| \begin{pmatrix} \frac{x+z}2  &  \frac{yt}{\sqrt2} & 0   \\  \frac{-yt}{\sqrt2}  &  \frac{x+z}2  & 0  \\  0  &  0  &  x  \end{pmatrix} \right\| 
 \le  \max  \left\{  \sqrt{ \left(\frac{   x  +  z}2 \right)^2 +  \frac{ t^2\,  y^2}  2 }  ,  |x| \right\} \\&
 \le  \max  \left\{  \sqrt{ \left(\frac{   |x|  +  |z|}2 \right)^2 +  \frac{ y^2}  2 }  ,  r \right\} .
\end{align*}
Let $s = \frac{|x|}r$ and $c = \frac{|y|}r$ and note that $s^2+c^2 \le 1$.
Hence
\begin{align*}
 \Big(\frac{ |x| + |z|}2 \Big)^2 +  \frac{ y^2}  2  &\le \frac{r^2}4 \big( (s+1)^2 + 2c^2 \big) 
 \le \frac{r^2}4 \big( s^2 + 2s + 1+ 2 - 2s^2 \big) \\&
 = \frac{r^2}4 \big( 4 - (1-s)^2 \big) \le r^2.
\end{align*} 
Therefore 
\[
 \|  \rho_t  (A)\| \le r \le \|A\|.
\]
Thus $\rho_t$ has norm 1 for every $t\in  [0,1]$.  
  
On the other hand,  $\wt{ \rho_t}$  is not  positive for every $t  > 1/\sqrt 2$ because
\[
  \wt {\rho_t}   \begin{pmatrix}      1    &    i    &  0   \\   -i    &  1  &  0  \\   0& 0&  0      \end{pmatrix}    
 = \frac 12 \begin{pmatrix}      1    &    i t \sqrt 2    &  0   \\   -i  t\sqrt 2    &  1  &  0  \\   0& 0&  0      \end{pmatrix}    .
\]  
We   also observe that  $\wt{ \rho_t}$ is positive for $t \le \frac1{\sqrt2}$. Indeed, if $A = \big[ a_{ij} \big] \ge 0$, 
\begin{align*} 
  \wt{ \rho_t}(A) = \wt{ \rho_t} \begin{pmatrix}  a_{11}  & a_{12}  & 0 \\ \ol{a_{12}}  & a_{22} & 0 \\  0  &  0  &  a_{33}  \end{pmatrix}   
 &=    \begin{pmatrix}  \frac{  a_{11 } +  a_{22}  +  2  a_{33}  }4  &  \frac{ i  {\im}   (a_{12} )   t}{ \sqrt 2}    &  0  \\            
 \frac{ -i  {\im}   (a_{12} ) t}{ \sqrt 2}     &    \frac{  a_{11 }+  a_{22}  +  2  a_{33}  }4   &  0   \\  0  &  0  & \frac{  a_{11 }+  a_{22} }2  \end{pmatrix} .
\end{align*} 
Since
\[
 \Big|  \frac{ i \, {\im}   (a_{12} )   \,  t}{ \sqrt 2}  \Big|^2 \le \frac{|a_{12}|^2}4 \le \frac{a_{11}a_{22} }4 \le \Big( \frac{a_{11} + a_{22}}4 \Big)^2 ,
\]
we obtain $\wt{ \rho_t}(A)  \ge 0$.
 
This example extends  to any dimension $n  >3$  by defining  
\[
 \sigma_{n,t}  \begin{pmatrix}  A  &   B   \\   C   &    D   \end{pmatrix}   :  =     \begin{pmatrix}  \rho_t(A)   &   0  \\  0  &  D \end{pmatrix}
\]
for $A \in M_3 (\R)$,  $B  \in  M_{2,n-3} (\R)$,  $C\in  M_{n-3,2} (\R)$ and $D  \in  M_{n-3,n-3}  (\R)$.  
The map $\sigma_{n,t}$ for $t \in(\frac1{\sqrt2},1]$ and $n>3$ is an example of a positive, unital, trace-preserving, norm-1 map on $M_n$ that has no positive complexification. 
\end{example}

\section{Witnesses of real entanglement}\label{sec:witnesses}

We now show that the existence of positive real maps with non-positive complexification is equivalent to the existence of quantum states that are entangled in real Hilbert space quantum mechanics, and yet are separable when regarded as states on  complex Hilbert spaces.     The equivalence is established by proving a duality between the cone of real positive maps and the cone of separable states in real Hilbert space quantum mechanics.

For $\K  \in  \{\R, \C\}$,  the cone of separable matrices is   

\[
 \SEP(\K^n\otimes \K^m)=
 \Big\{   \sum_{i} A_i\otimes B_i  :  A_i\in \PSD(\K^n), \,  B_i\in \PSD(\K^m) \Big\} .
\]
The elements of $\SEP(\K^n\otimes \K^m)$ will be called {\em $\K$-separable}, and the elements of   
$\PSD (\K^n\otimes \K^m) \setminus   \SEP(\K^n\otimes \K^m)$ will be called {\em $\K$-entangled}.  Note that the separable bipartite elements over various fields (real, complex and quaternions) have been considered before (see \cite{Hildebrand}). We take a closer look at the real case and consider various level of separability in terms of Schmidt rank.
So more generally, we consider the cone of positive matrices  with Schmidt rank at most $p$ \cite{terhal2000schmidt}.
We call these matrices  {\em $\K$-$p$-separable}  and denote them by 
\[
 \SEP_p(\K^n\otimes \K^m) =
 \Big\{ \sum_{j}  \Vect(A_j)  \Vect(A_j)^*  :  A_j  \in  M_{n,m}  (\K) \, , \rank ( A_j) \le p \Big\} .
\]

Equivalently, $\SEP_p(\K^n \otimes \K^m)$ is the cone generated by the set of all positive matrices of the form $uu^*$ for $u \in \K^n \otimes \K^m$ that can be expressed as a sum of $p$ or fewer elementary tensors. 

Note that $ \SEP_1(\K^n\otimes \K^m)  =   \SEP(\K^n\otimes \K^m)$. 
The elements of   $\PSD (\K^n\otimes \K^m) \setminus   \SEP_p(\K^n\otimes \K^m)$ will be called {\em $\K$-$p$-entangled}. 
 Note also that for $p \ge min \{n,m \}$ we have that $\SEP_p(\K^n \otimes \K^m) = \PSD(\K^n \otimes \K^m)$.

We also define the cone consisting  of all real positive  matrices that are   $\C$-$p$-separable, denoted as
\begin{align*}
 \CSEP_p(\R^n \otimes \R^m)   &:= \SEP_p(\C^n \otimes \C^m) \cap \PSD(\R^n \otimes \R^m)   \\ 
   &  = \Big\{   \sum_{j}    \Vect(A_j)  \Vect(A_j)^* +  \Vect(\overline A_j)  \Vect(\overline A_j)^*   : \\[-2.0ex] 
     &  \qquad   \qquad  A_j  \in  M_{n,m}  (\C)  \AND \rank ( A_j) \le p\,  \FORAL  j    \Big\} .
\end{align*}
For $p=1$, we use the notation  $ \CSEP(\R^n \otimes \R^m) :  =  \CSEP_1(\R^n \otimes \R^m) $. 

We now provide an equivalent characterization of $ \CSEP_p(\R^n \otimes \R^m)$, which shows that every $\C$-$p$-separable real matrix is also $\R$-$2p$-separable:

\begin{prop}\label{prop:csepp}   One has   
\begin{align*}
 \CSEP_p(\R^n \otimes \R^m)     &  = \Big\{   \sum_{j}    \Vect(X_j)  \Vect(X_j)^t +  \Vect(Y_j)  \Vect(Y_j)^t   : \\[-2ex]
     &  \qquad \qquad     X_j, Y_j    \in  M_{n,m}  (\R)    , \,  \rank ( X_j+  i  Y_j  ) \le p\, , \forall j   \Big \}   
\end{align*}
and $\CSEP_p(\R^n \otimes \R^m) \subset \SEP_{2p}(\R^n \otimes \R^m)$. 
\end{prop}

\begin{proof}
The characterization follows from decomposing  $A_j$ as $A_j  =  X_j   +  i  Y_j$,  with $X_j,  Y_j  \in  M_{m,n}  (\R)$,  and from the relation  
\[
 \Vect(A_j)  \Vect(A_j)^* +  \Vect(\overline A_j)  \Vect(\overline A_j)^* =   \Vect(X_j)  \Vect(X_j)^t +  \Vect(Y_j)  \Vect(Y_j)^t   .
\]  
The inclusion equality follows from the first and from the fact that 
$\rank  (A_j) \le p$ implies  $\rank  (X)\le 2p$ and $\rank(Y) \le 2p$.  
\end{proof}

In the next section we will provide an example showing that the inclusion $\CSEP_p(\R^n \otimes \R^m) \subset \SEP_{2p}(\R^n \otimes \R^m)$ is strict.

\begin{rem}
Note that the cone  $\SEP_p(\K^n \otimes \K^m)$ is closed.
To see this it is enough to show that the convex set consisting of the elements in these cones of trace 1 is a closed convex set. 
But the elements in $\SEP_p(\K^n \otimes \K^m)$ of trace one is the convex hull of the set of  the matrices of the form $A =  V V^*$ where  $V$ is a unit vector with Schmidt rank $\le p$. 
The set of such matrices is compact  and, in finite dimensions, 
the convex hull of a compact set is again a compact set by Caratheodory's Theorem. 
The cone $\CSEP_p(\R^n \otimes \R^m)$ is also closed since it is the intersection of two closed sets.
\end{rem}

The complex version of part (1) of the next theorem was proven in \cite{skowronek2009cones}.

\begin{thm}\label{dualcones} 
Let $\Phi : M_n(\R) \to M_m(\R)$ be a  map that commutes with the adjoint and let $C_{\Phi}$ be its Choi matrix. Then,
\begin{enumerate}
\item[(i)] $\Phi$ is $p$-positive if and only if
\[
 \Tr (C_{\Phi} \,  P) \ge 0 \qforal P \in \SEP_p(\K^n \otimes \K^m),
\]

\item[(ii)] The complexification $\wt{\Phi}:M_n(\C) \to M_m(\C)$ is  $p$-positive if and only if
\[
 \Tr (C_{\Phi} \,P) \ge 0 \qforal P \in \CSEP_p(\R^n \otimes \R^m).
\]
\end{enumerate}
\end{thm}

\begin{proof} 
The set $\SEP_p(\K^n \otimes \K^m)$ is the convex hull of its rank-one  elements, of the form $ \Vect (A) \Vect  (A)^*$ with $\rank  (A) \le p$.  
The condition   $\Tr (C_{\Phi}P) \ge 0 \qforal P\in \SEP_p(\K^n \otimes \K^m)$ is equivalent to  
\[
 0  \le    \Tr ( C_\Phi  \,   \Vect (A) \Vect  (A)^*  )  = \<  \Vect (A)|  C_\Phi   \Vect  (A)\>
\]
for all matrices with rank at most $p$.  
In turn, this condition is equivalent to $\<  V  |  C_\Phi  V\>  \ge0$ for every vector $V$ with Schmidt number at most $p$.   
By Lemma \ref{lemma-pos}, this condition is equivalent to $p$-positivity of $\Phi$.  This proves (i).      

We now  prove (ii).   
By Proposition \ref{prop:csepp},    the condition   $\Tr (C_{\Phi}  P) \ge 0$ for all  $P \in \CSEP_p(\R^n \otimes \R^m)$ is equivalent to  
\begin{align*}
 0    &\le    \Tr   ( C_\Phi  \,     (  \Vect (X) \Vect  (X)^t  +  \Vect (Y) \Vect  (Y)^t   )  )    \\
   &   = \<  \Vect (X)|  C_\Phi   \Vect  (X)\>  +   \<  \Vect (Y)|  C_\Phi   \Vect  (Y)\> \, ,
\end{align*} 
for all real matrices  $X,Y$  such that $\rank (X+  i Y)\le p$.  
By Theorem \ref{theo:ChoiCondition}, this condition is equivalent to $p$-positivity of $\widetilde \Phi$. 
\end{proof}

\begin{cor}\label{cor:duality} 
Let $P\in \PSD(\K^n \otimes \K^m)$. Then:
\begin{enumerate}
\item[(i)] $P$ is $\K$-$p$-separable if and only if $\Tr (C_\Phi P)\geq 0$, for all $p$-positive $\Phi:M_n(\K)\rightarrow M_m(\K),$
\item[(ii)] for $\K  = \R$,  $P$ is $\C$-$p$-separable if and only if $\Tr (C_{\Phi}P) \geq 0$ for all $\Phi: M_n(\R) \to M_m(\R)$ such that $\wt{\Phi}$ is $p$-positive.
\end{enumerate}
\end{cor}

\begin{proof} 
The proof uses the fact that these cones are closed and apply the standard Hahn-Banach results for separating a point from a closed cone.
\end{proof}

This corollary implies that, if a real positive matrix $P$  is $\R$-$p$-entangled, then there must exist
a $p$-positive map $\Phi$ such that $\Tr (C_\Phi P)<0$.
We say that $\Phi$  is a \emph{witness of  $\R$-$p$-entanglement} for $P$.
Similarly, if a real positive matrix $P$ is $\C$-$p$-entangled, then there must exist a positive map $\Psi$ with positive complexification  such that $\Tr (C_\Psi P)<0$. 
We say that $\Psi$ is  a {\em witness of $\C$-$p$-entanglement} for  $P$.
In the literature, the term entanglement witness was typically used for $p=1$, while  witnesses for $p>1$ were sometimes called Schmidt number witnesses  \cites{terhal2000bell,sanpera2001schmidt}.

Before proceeding to the next result we need a lemma for which we refer to \cite{CMW}. The assertion holds for any field.
\begin{lem}\label{schmidt-dec}
Every vector $V\in  \K^n\otimes \K^m$  with Schmidt rank $\le p$ can be written as $$V =   (I_n  \otimes S) W=(T\otimes I_m)Z ,$$ where $W$ is a vector in $  \K^n\otimes \K^p$, $Z$ is a vector in $\K^p\otimes \K^m$ with  $S:  \K^p  \to  \K^m$  and $T:  \K^p  \to  \K^n$ are partial isometries.
\end{lem}

The fact that $p$-positive maps are witnesses of $p$-entanglement is equivalently expressed by the following corollary,  the complex version of which was established  in \cite{Horodeckis} for  $p=1$ and in \cite{terhal2000schmidt} for general $p$.

\begin{cor}\label{horodeckis}
For a matrix $P\in \PSD(\K^n \otimes \K^m)$, the following are equivalent 
\begin{enumerate}
\item[(i)] $P$ is $\K$-$p$-separable,
\item[(ii)] $\Phi \otimes \id _m(P) \geq 0$ for every  $r\in  \N$ and for every $p$-positive map $\Phi: M_n(\K) \to M_r(\K)$,
\item[(iii)] $\Phi \otimes \id _m(P) \geq 0$ for every $p$-positive map $\Phi: M_n(\K) \to M_m(\K)$.
\item[(iv)] $\id_n \otimes \Phi(P)\geq 0$ for every $p$-positive map $\Phi:M_m(\K)\rightarrow M_l(\K)$
\end{enumerate}
\end{cor}

\begin{proof}

(i)  $\Rightarrow$ (ii).  If $P$ is $\K$-$p$-separable, it can be written as $P  =  \sum_i   V_i  V_i^*$, 
where each $V_i$ is a vector in $\K^n\otimes \K^m$ with Schmidt rank at most $p$.    
Using Lemma \ref{schmidt-dec} we know that every vector $V_i \in  \K^n\otimes \K^m$  with Schmidt rank $\le p$ can be written 
as $V_i  =   (I_n  \otimes S_i ) W_i$, where $W_i$ is a vector in $  \K^n\otimes \K^p$ and  $S_i:  \K^p  \to  \K^m$ is a partial isometry.
Then, for every map $\Phi: M_n(\K) \to M_r(\K)$, we have 
\[
 \Phi \otimes \id _m(P) =  \sum_i   \,   (I_r\otimes   S_i)        ~(\Phi \otimes \id _p(   W_i W_i^*  )) ~  (I_r\otimes   S_i^*)\ge 0 .   
\]  
If $\Phi$ is $p$-positive, then each term in this sum is positive; so $\Phi \otimes \id _m(P) \ge 0$.       

(ii) $ \Rightarrow$ (iii).  Take $ r  = m$.  

(iii) $ \Rightarrow$ (i).  Given a map $\Phi$, the condition   $\Phi \otimes \id _m(P) \geq 0$ implies   
\begin{align*}
 0  &\le     \sum_{i,j}   \Tr  ( ( E_{i,j}  \otimes E_{i,j})    \Phi \otimes \id_m  (P) )   
 =     \sum_{i,j}   \Tr  (( \Phi^*( E_{i,j})  \otimes E_{i,j}) \, P ) \\
 & =     \Tr  ( S C_{\Phi^*}  S^* P )   =  \Tr  ( C_{\Phi^*}  S^* P  S) , 
\end{align*} 
where $S : \K^m \otimes \K^n \to \K^m \otimes \K^n$ is the linear operator defined by $S( w\otimes v) =  v\otimes w$ for all  $v\in  \K^n$ for all $w\in  \K^m$.   
If $\Phi:  M_n  (\K) \to  M_m(\K)$ is an arbitrary $p$-positive map, then $\Phi^*:  M_m  (\K) \to  M_n(\K)$ is an arbitrary $p$-positive map by Lemma \ref{lemma-pos}. 
For $\K  = \R$,  Corollary \ref{cor:duality} implies that $S^*PS$, and therefore $P$, is $\R$-$p$-separable. 

The equivalence of (i) and (iv) follows exactly like the equivalence of $(1)$ and $(3)$ using the second equality of Lemma \ref{schmidt-dec}.
\end{proof}

\begin{cor}
For every $p\in  \N$, the following are equivalent
\begin{enumerate}
\item[(i)]  there exists a $p$-positive map $\Phi:  M_n  (\R) \to M_m (\R)$ with non-$p$-positive complexification,
\item[(ii)] there exists an $\C$-$p$-separable matrix  $P   \in  \PSD (  \R^n\otimes \R^m)$ that is not $\R$-$p$-separable.    
\end{enumerate}
\end{cor}

\begin{proof}
Immediate from Corollary \ref{cor:duality}. 
\end{proof}

\begin{example}
 By \cite{choi:poslin} the map $\Phi: M_n(\R) \to M_n(\R)$ given by 
 \[
  \Phi(X) = (n-1) \Tr(X) I_n - X
 \]
is $(n-1)$-positive but not $n$-positive. 
Its complexification  $\tilde{\Phi}$ is the map on $M_n(\C)$ given by the same formula, and it is also $(n-1)$-positive but not $n$-positive. 
By the above dualities, we see that
\[
 \CSEP_{(n-1)}(\R^n \otimes \R^n) \subsetneq \PSD(\R^n \otimes \R^n) = \CSEP_n(\R^n \otimes \R^n).
\]
\end{example}

\section{Examples of $\R$-entangled states that are $\C$-separable}\label{sec:examples2}  

In this section, we use the distinctions between the various cones of positive maps to provide explicit examples. 
The first example is an $\R$-$(2p-1)$-entangled but $\C$-$p$-separable matrix.
The remaining examples focus on the $p=1$ case, providing examples of quantum states that are entangled in 
real Hilbert space quantum mechanics, but separable in the complex version. 
The key observation is that real separable states must be invariant under partial transposition (IPT)  \cite{caves2001entanglement}.

\begin{example}\label{ex:pentangled}
Here we provide an example of a matrix $P  \in   \PSD ( \R^{2p} \otimes \R^{2p})$  that is $\R$-$(2p-1)$-entangled, but $\C$-$p$-separable.  The matrix is 
\[
 P  =    \frac{V V^*   +   \overline V  V^t}2 , 
\] 
where  $V  \in   \C^{2p}  \otimes  \C^{2p} $ is the complex vector    $V   =   \sum_{j=1}^p   \alpha_j \otimes \alpha_j/\sqrt p$, 
with $\alpha_j   =  (  e_j^{(2p)}  +i  \, e_{j+p}^{(2p)})/\sqrt 2$.    
It is evident from the definition  that $P$ is $\C$-$p$-separable.    On the other hand,  $P$ can be equivalently rewritten as  
\[
 P  =  \frac{ \Vect(  O_+)   \Vect(  O_+)^t  +   \Vect(O_-)   \Vect(  O_-)^t}{4p}  ,
\]
with $O_+   = \begin{spmatrix}   0  &  I_p  \\  I_p  &  0  \end{spmatrix}$ and $O_-   = \begin{spmatrix}    I_p  &  0  \\   0&   -I_p   \end{spmatrix} $.  
If $P$ were $\R$-$(2p-1)$-separable, there should exist at least one real vector with Schmidt rank $\le 2p-1$  in the linear span of $\Vect(O_+)$ and $\Vect (O_-)$.  
But this is not possible, because the real span of $O_+$ and $O_-$ contains only full-rank matrices. 
\end{example}

The following fact has been observed in \cite{caves2001entanglement}. 
\begin{lemma}\label{RSEP is IPT}
Every  $P  \in   \SEP  (\R^n\otimes \R^m)$ must satisfy the condition $P   =     \tau_n \otimes \id_m  (P)  =  \id_n \otimes \tau_m  (P)$, where  $\tau_r:   M_r  (\R)  \to  M_r  (\R)  ,  \,    A \mapsto  A^t $ is the transpose map on $M_r  (\R)$.    
\end{lemma}

Invariance under partial transpose is a necessary condition for separability of real quantum states, similarly to the positive partial transpose (PPT)  criterion for separability in  the complex case \cites{peres1996separability,Horodeckis,horodecki1997separability}.  Based on this analogy, we call the invariance under partial transpose the IPT criterion for separability of real quantum states.  Note that, like the PPT criterion, the IPT criterion is generally not sufficient for separability, neither on the reals nor on the complexes.  Counterexamples can be  constructed in a standard way, starting from examples of unextendible product bases (UPBs), that is, sets of orthonormal product vectors in $\C^n\otimes \C^m$ with the property that their orthogonal complement contains no product vector \cite{bennett1999unextendible}.  For a UPB, the projector on the orthogonal complement is  entangled (on the complex domain) and has  positive partial transpose. If  the UPB  consists of vectors with real entries, then the projector on the complement is entangled   and invariant under partial transpose. Hence, it provides an example of an IPT matrix that is neither in $\SEP  (\R^n\otimes \R^m)$ nor in $\CSEP  (\R^n\otimes \R^m)$.   Explicit examples of UPBs with real entries can be found in  \cite{bennett1999unextendible}.


\begin{rem} \label{P:strict inclusion}
The above criteria gives us an easy way to see that for $n,m \ge 2$ the inclusions
\[
 \SEP(\R^n \otimes \R^m) \subsetneq \CSEP(\R^n \otimes \R^m),
\]
are strict.

For $n=m=2,$ consider the matrix 
$A=\begin{pmatrix}
0 & -i\\
i & 0
\end{pmatrix}$
and consider the Hermitian operator
\[
 P = I_2\otimes I_2+ A\otimes A
 = \left(   \begin{array}{cc|cc}  1 & 0 & 0 & -1\\  0 & 1 & 1 & 0\\ \hline 0 & 1 & 1 & 0\\ -1 & 0 & 0 &1  \end{array} \right)
  \in \PSD(\R^2\otimes \R^2)
\]
Since $I_2 \pm A \ge 0$ and
\[
 2P = (I_2 +A)\otimes (I_2+A) + (I_2-A) \otimes (I_2-A) ,
\]
$P$ belongs to $\CSEP(\R^2 \otimes \R^2)$. 
However, $\id \otimes T(P)\neq P$ and hence $P\not \in  \SEP(\R^2 \otimes \R^2)$.
 For the general case, one can  form matrices $A_n= A \oplus 0_{n-2} \in M_n(\C)$ and 
$A_m= A \oplus 0_{m-2} \in M_m(\C)$ and observe that the same argument shows that
$I_n \otimes I_m + A_n \otimes A_m$ is $\CSEP(\R^n \otimes \R^m)$ but is not $\R$-separable. 
\end{rem}

Other examples of $\C$-separable maps that are not equal to their partial transpose (and therefore are $\R$-entangled) are provided in the following.

\begin{example}
Consider the real positive matrix 
\begin{align*}
 P_s &=  \frac{s}{n(n+1)}(I_n\otimes I_n+W) + \frac{1-s}{n(n-1)}(I_n\otimes I_n-W) \\
 &= \frac{n\!+\!1 - 2s}{n^3-n} \Big( I_n\otimes I_n + \big(2ns - (n\!+\!1)\big) W \Big) \qfor  s\in  [0,1] \, , 
\end{align*}
where $W=\sum_{i,j}E_{i,j}\otimes E_{j,i}\in M_n(\C)\otimes M_n(\C)$. 
It is known \cite{werner-89} (see also chapter 6 of \cite{watrous}) that  $P_s$ is $\C$-separable for $s  \in  [\frac{1}{2},  1]$.  
On the other hand,  $P_s$ is not equal to its partial transpose unless $s=  \frac{n+1}{2n}$, 
and therefore $P_s$  is  $\R$-entangled for $s$ in $\big[ \frac 12, 1 \big] \setminus \big\{ \frac{n+1}{2n} \big\}$. \\
\end{example}

The above example can be generalized as follows:

\begin{prop}
For every  symmetric matrix $A\in   M_n(\R)\otimes M_m(R)$ with $\id \otimes T(A)\neq A$, there exists an  $s_* >  0$ such that  the matrix
\begin{align*}
 P_s =(1-s) I_n \otimes I_m + s\,  A
\end{align*}
is $\R$-entangled and $\C$-separable for every $s\in  (0,s_*]$.  
\end{prop}

\begin{proof}   
A  result of Gurvits and Barnum \cite{gur-bar}  states that, for every Hermitian matrix $H\in \Herm(\C^n\otimes \C^m)$ with  $\|H\|_2\leq 1$, 
the matrix   $I_n\otimes I_m +H$ is $\C$-separable.   
If $ A$ is a real symmetric matrix, then this result implies that the matrix $I_n\otimes I_m + \frac{s}{1-s} A$ is $\C$-separable whenever $s/(1-s)  \le \|  A\|_2$.
Therefore the matrix $P_s$ is $\C$-separable whenever $s \le  s_*:  =\|A \|_2/(\|  A\|_2  + 1)$.     
On the other hand,  $\id \otimes T(P_s)\neq P_s$ for every $s>0$, and therefore $P_s$ is $\R$-entangled for every $s>0$. 
\end{proof}

\section{Entanglement $p$-breaking maps}\label{sec:breaking}

In \cite{HSR},  Horodecki, Shor, and Ruskai  introduced the notion of {\em entanglement breaking} maps, 
that is, maps that transform every entangled state into a separable state when acting on one of the components.   
In general, we say that a  map $\Phi: M_n(\K)\rightarrow M_m(\K)$ is  \emph{$\K$-entanglement $p$-breaking} if 
\begin{align*} 
  \Phi\otimes \id _r(P) \in \SEP_p(\K^m\otimes \K^r) \qforal P \in \PSD(\K^n\otimes \K^r) \AND  r \ge1.
\end{align*}
For $p=1$, we just call the map $\Phi$ $\K$-entanglement breaking.   
In the complex case, entanglement $p$-breaking maps had been considered in \cite{skowronek2009cones}, where they were called $p$-superpositive maps.      

The cone of $\K$-entanglement breaking ($\K$-entanglement $p$-breaking) maps $\Phi: M_n(\K)\rightarrow M_m(\K)$  will be denoted  as  $\K\text{-}\EB (n,m)$  ($\K\text{-}\EB_p  (m,n)$).  We now provide a characterization of the $\K$-entanglement $p$-breaking channels. 
Note that since the cone $\SEP_p(\K^m\otimes \K^r)$ is larger for larger $p$, meaning that $\K$-entanglement $p$-breaking is a weaker condition than $\K$-entanglement breaking when $p\ge2$.

In the following result, the complex versions of (1)-(3) and  (4) were provided for $p=1$ in \cite{HSR} and \cite{JKPP}, respectively. 
The complex version of (5) and (6) was provided in \cite{skowronek2009cones} for general $p$.

\begin{theo}\label{Rentanglementbreaking characterization} 
Let $\Phi: M_n(\K) \to M_m(\K)$ be a linear map. Then the following are equivalent.
\begin{enumerate}
 \item[(i)] $\Phi$ is $\K$-entanglement $p$-breaking,
 \item[(ii)] $C_\Phi \in \SEP_p(\K^n \otimes \K^m)$,
 \item[(iii)] there exist matrices $(C_i)_{i=1}^k  \subset  M_{m,n}(\K)$ such that  $\rank  (C_i) \le p$ for every $i$ and
 \[
  \Phi(X) = \sum_{i=1}^k C_i X  C_i^*, 
 \]
 \item[(iv)] $\Phi = \Delta \circ \Gamma$ where $\Gamma: M_n(\K) \to l^\infty_k(\K)  \otimes M_p  (\K)$ and $\Delta : l^\infty_k(\K) \otimes M_p  (\K)\to M_m(\K)$ are completely positive maps for some $k\ge1$.
 \item[(v)] for every $r$ and every $p$-positive map $\Psi: M_m(\K) \to M_r(\K)$ the map $\Psi \circ \Phi$ is completely positive,
 \item[(vi)] for every $r$ and every $p$-positive map $\Psi: M_r(\K) \to M_n(\K)$ the map $\Phi \circ \Psi$ is completely positive.
\end{enumerate}
\end{theo}

\begin{proof}
(i) $\Rightarrow$ (ii): Applying $\id_n \otimes \Phi$ to $X=  \sum_{i,j}  E_{ij}  \otimes E_{ij}$, we obtain  that   
$C_\Phi = (\id_n \otimes \Phi) \big( \sum_{ij} E_{ij}\otimes E_{ij} \big)$ belongs to  $ \SEP_p(\K^n \otimes \K^m)$. 

(ii) $\Rightarrow$ (iii): If  $C_\Phi  \in  \SEP_p(\K^n \otimes \K^m)$, then it can be written as 
$C_\Phi   =  \sum_{i=1}^k     \Vect(A_i)  \Vect( A_i)^*$, with $\rank (A_i)  \le p$ for every $i$.  
Using Eq. (\ref{choiinverse}), we obtain 
\begin{align*} 
  \Tr  [\Phi  (    X )  \, Y]    &= \Tr  (   C_\Phi   (X^t \otimes  Y))     \\&
  =    \sum_{i=1}^k    \<  \Vect(A_i) ,   (X^t\otimes Y)  \Vect( A_i)\>  \\&
  =     \sum_{i=1}^k   \Tr  (A_i^t    X \overline A_i   Y)\, , 
\end{align*} 
for arbitrary matrices $X$ and $Y$.  Since  $Y$ is arbitrary, we conclude $\Phi  (X)   =  \sum_{i=1}^k  C_i  X  C_i^*$ with $C_i  : =  A_i^t$. 

(iii) $\Rightarrow$ (i):     If   $\Phi  (X)   =  \sum_{i=1}^k  C_i  X  C_i^*$ with $\rank  (C_i) \le p$, then for every $A  \in   M_{n,m}(\K)$, one has    
\[
 (\Phi \otimes \id_n)   (\Vect (A)  \Vect(A)^*) =  \sum_i  \Vect  (C_i  A)   \Vect ( C_i  A)^* .   
\]
Since $\rank  (C_i  A)  \le \rank  (C_i) \le p$, we deduce that  $(\Phi \otimes \id_n)   (\Vect (A)  \Vect(A)^*)$   is $\K$-$p$-separable.        
Since $A$ is arbitrary,  we conclude that  $(\Phi \otimes \id_n)   V V^* $   is $\K$-$p$-separable for every vector $V$,  
and therefore $ (\Phi \otimes \id_n) (X)$ is   $\K$-$p$-separable for every positive $X$.   
This proves the equivalence of (1), (2), and (3).

(iii)$ \Rightarrow$ (iv):   Since $\rank  (C_i)\le p$, there exists an isometry  $S_i:   \ran  (C_i) \to   \K^p$.
 Let $\delta_1,\dots,\delta_k$ be the standard basis for $l^\infty_k(\K)$. Define
\[
 \Gamma(X) = \sum_{i=1}^k  \delta_i  \otimes  S_i C_i X  C_i^*  S_i^*  \qand \Delta\big( \sum_{i=1}^k   \delta_i \otimes Y_i\big) = \sum_{i=1}^k   S_i^* Y_i  S_i .
\]
The form of $\Gamma$ and $\Delta$ guarantees  that they are completely positive. By construction, $\Phi = \Delta \circ \Gamma$.

(iv) $\Rightarrow$ (v):  The action of  $\Delta$  on a generic element of $l^\infty_k  (\K)\otimes M_p$ can be expressed as 
$\Delta  ( \sum_i  \delta_i \otimes  Y_i )  =  \sum_i       \Delta_i   (Y_i)$, with $\Delta_i:  M_p   (\K) \to M_m(\K)$ completely positive.   
Since the domain of  $\Delta_i$   is $M_p  (\K)$,  $\Delta_i$ is $\K$-entanglement $p$-breaking.  
Hence $\Delta$ is also $\K$-entanglement $p$-breaking. 
Using the implication  $(1) \Rightarrow (2)$, we conclude that each $C_{\Delta_i}$ is $\K$-$p$-separable.   
For an arbitrary $p$-positive map  $\Psi :  M_m(\K) \to M_r  (\K)$,  $(\Psi\circ \Delta_i)$ is completely positive.    
Indeed,    $C_{  \Psi\circ\Delta_i}    =  (\id_p  \otimes \Psi)   C_{\Delta_i}$, and Corollary  \ref{horodeckis} implies that   $(\id_p  \otimes \Psi)   C_{\Delta_i}$ is positive whenever $C_{\Delta_i}$ is $\K$-$p$-separable.   
We conclude that  $\Psi\circ\Phi$ is completely positive, as it factors as $\Psi\circ\Phi = (\Psi\circ \Delta)\circ\Gamma$, which is the composition of two completely positive maps.     

(v) $\Rightarrow$ (vi):   If  $\Psi \circ \Phi$ is completely positive, then  $C_{\Psi\circ \Phi}    =   ( \id_m\otimes \Psi)  (C_\Phi)\ge 0$.    
If $\Psi$ is an arbitrary $p$-positive map, this implies that $C_\Phi$ is $\K$-$p$-separable, by Corollary \ref{horodeckis}.   
Hence, also $C_{\Phi^*}$ is $\K$-$p$-separable. 
Using the implication  $(2)\Rightarrow (1)$, we obtain that $\Phi^*$ is $\K$-entanglement $p$-breaking.  
Using the implication  $(1)\Rightarrow (5)$,  we conclude that  $   \Psi  \circ \Phi^*$ is completely positive for every $p$-positive $\Psi$.  
Now,  a generic map $\Gamma$  is completely positive if and only if its adjoint  $\Gamma^*$ is completely positive  (by Lemma \ref{lemma-pos}). 
Hence, the map $\Phi  \circ \Psi^*$ is completely positive for every $p$-positive $\Psi$.  
But for an arbitrary $p$-positive map $\Psi$, $\Psi^*$ is an arbitrary $p$-positive map (by Lemma \ref{lemma-pos})).  
Hence,  $\Phi  \circ \Psi$ is completely positive for every $p$-positive $\Psi$.

(vi) $\Rightarrow$ (ii):   For every map $\Psi$, the map $\Phi \circ \Psi$  is completely positive if and only if its adjoint $\Psi^*  \circ \Phi^*$ is completely positive.   
If $\Psi$ is an arbitrary $p$-positive map, then $\Psi^*$ is an arbitrary $p$-positive map (by Lemma \ref{lemma-pos}).
Therefore  $C_{\Phi^*}$ is $\K$-$p$-separable, by Corollary \ref{horodeckis}.   
Recalling Eq. (\ref{choiadjoint}),  we conclude that $C_\Phi$ is $\K$-$p$-separable.
\end{proof}

 We state the analogous theorem for $\CSEP$ without proof.

\begin{theo}\label{CSEP-breaking}
Let $\Phi: M_n(\R) \to M_m(\R)$ be a linear map. Then the following are equivalent:
\begin{enumerate}
\item[(i)] $\Phi \otimes id_r(P) \in \CSEP_p(\R^n \otimes \R^r)$, for every $r$ and every $P \in \PSD(\R^n \otimes \R^r)$,
 \item[(ii)] $\tilde{\Phi}$ is $\C$-entanglement $p$-breaking,
 \item[(iii)] $C_\Phi \in \CSEP_p(\K^n \otimes \K^m)$,
 \item[(iv)] there exist matrices $(C_i)_{i=1}^k  \subset  M_{m,n}(\C)$ such that  $\rank  (C_i) \le p$ for every $i$ and
 \[
  \Phi(X) = \sum_{i=1}^k C_i X  C_i^*, 
 \]
 for every $X \in M_n(\R)$,
 \item[(v)] there exists  $k \ge 1$, a real C*-subalgebra $\mathcal{C} \subseteq \ell^{\infty}_k(\C) \otimes M_p(\C)$ and a factorization, $\Phi = \Delta \circ \Gamma$ where $\Gamma: M_n(\R) \to \mathcal{C}$ and $\Delta : \mathcal{C} \to M_m(\R)$ are  real linear completely positive maps,
 \item[(vi)] for every $r$ and every $p$-positive map $\Psi: M_m(\C) \to M_r(\C)$ the map $\Psi \circ \tilde{\Phi}$ is completely positive,
 \item[(vii)] for every $r$ and every $p$-positive map $\Psi: M_r(\C) \to M_n(\C)$ the map $\tilde{\Phi} \circ \Psi$ is completely positive.
\end{enumerate}
\end{theo}

Here by a {\it real C*-subalgebra} of a C*-algebra, we mean a closed subset that is a real subalgebra that is also *-closed.
We recall that a real C*-subalgebra of $\ell^{\infty}_k(\C) \otimes M_p(\C)$ need not be isomorphic to $\ell^{\infty}_j(\R) \otimes M_p(\R)$ for some $j$, even in the case $p=1$.

\begin{example}
Here we provide an example of a map  $\Phi : M_{2p}(\R) \to M_{2p}(\R)$ that has a $\C$-entanglement $p$-breaking complexification, 
but is not $\R$-entanglement $(2p-1)$-breaking (and therefore is also not entanglement $p$-breaking).    The map is 
\[
 \Phi  (A) =        O_+   A  O_+  +  O_-  A  O_-  \, ,
\]
where  $O_\pm$ are the orthogonal matrices  
$O_+   = \begin{spmatrix}   0  &  I_p  \\  I_p  &  0  \end{spmatrix}$ and $O_-   = \begin{spmatrix}    I_p  &  0  \\   0&   -I_p   \end{spmatrix} $.  
Its Choi matrix $C_{\Phi}$ was studied in example \ref{ex:pentangled}, where we showed that $C_\Phi$ is $\C$-$p$-separable, but $\R$-$(2p-1)$-entangled.  
By Theorem \ref{Rentanglementbreaking characterization}, this implies that $\widetilde\Phi$ is $\C$-entanglement $p$-breaking, while $\Phi$ is not $\R$-entanglement $(2p-1)$-breaking.  
\end{example}

For $p=1$, the above example shows that a completely positive map can break $\C$-entanglement without breaking $\R$-entanglement.

\section{IPT states and real analogues of the PPT-squared  conjecture}\label{sec:PPT}
 In this section we discuss  a conjecture, which we believe might be the best real-vector-space analogue of the so called PPT-squared conjecture. The connection between transpose invariant states and real separable states is crucial in our analysis in this section. So we note down this connection in the following.
\subsection{IPT maps and real separable states}

For $\K  =  \C$ or $\R$,  a positive matrix $P  \in  \PSD  (\K^n\otimes \K^m)$ is called 
\begin{itemize}
\item  {\em PPT  $($positive partial transpose$)$}  if   $ (   \tau_n\otimes \id_m)    (P)  \ge  0$,
\item  {\em IPT  $($invariant under partial transpose$)$}  if   $ (   \tau_n\otimes \id_m)    (P)  =  P$.  
\end{itemize}
Hereafter, the set of IPT matrices will be denoted by $\IPT(\K^n\otimes\K^m)$.

A  completely positive map   $\Phi: M_n(\K) \to M_m(\K)$ is called PPT (IPT)   if its  Choi matrix is PPT  (IPT).  The set of IPT maps can be equivalently characterized as the set of  maps for which the Choi matrix $C_\Phi$ and the Jamio\l kowski matrix $  J_\Phi   =   (\tau_n\otimes \id_m)   (C_\Phi)$ coincide and are positive.

Clearly, every entanglement breaking map is PPT:  if the Choi matrix is separable,   then it is also PPT by the Peres-Horodecki criterion \cites{peres1996separability,Horodeckis,horodecki1997separability}. 
In the real case, entanglement breaking maps and their complexifications are IPT.  The converse is generally not true:   the existence of entangled IPT matrices implies the existence of IPT maps that are not entanglement-breaking.

\begin{prop} \label{C:RPPT equals transpose}
If $\Phi: M_n(\R) \to M_m(\R)$ is $\R$-entanglement breaking, then its complexification $\widetilde \Phi$ is  IPT, 
and   $\tau_m \circ \widetilde \Phi =  \widetilde \Phi  \circ \tau_n  = \widetilde \Phi$. 
\end{prop}

\begin{proof} 
Since $C_{\Phi} \in \SEP(\R^n \otimes \R^m)$, we can write $C_\Phi = \sum_k A_k \otimes B_k$ for $A_k \in \PSD(\R^n)$ and $B_k\in \PSD(\R^m)$.
Recalling that $C_{\widetilde \Phi}  = C_\Phi$, we obtain
\[
 \tau_n \otimes \id_m(C_{\widetilde \Phi}) = \sum_k A_k^t \otimes B_k = \sum_k A_k \otimes B_k = C_{\widetilde \Phi } \, , 
 \]
that is  $\widetilde \Phi$ is IPT.    Moreover, one has
\[
 C_{\tau_m \circ \widetilde \Phi} =  \id_m\otimes  \tau_m (C_{\widetilde \Phi})  =   \sum_k A_k \otimes B_k^t = \sum_k A_k \otimes B_k = C_{\widetilde \Phi} \,,
\]  
which implies $\tau_m \circ \widetilde \Phi  =  \Phi$, and 
and 
\[  
 C_{ \widetilde \Phi\circ \tau_n} =  C_{ \tau_m\circ \widetilde \Phi\circ \tau_n}  
 =     \sum_{i,j}    E_{i,j}  \otimes   (\widetilde \Phi   ( E_{i,j}^t))^t =   \sum_{i,j}    E_{i,j}  \otimes   (\Phi   ( E_{i,j}^t))^t =  C_{\Phi} \, ,      
\]
where the last equality follows from the fact that  $\Phi$ commutes with the adjoint.  
\end{proof}

Interestingly, the real separable states can be identified as those states in $\CSEP$ which have the IPT property:

\begin{prop}\label{IPT+CSEP=RSEP}
The following equality of sets holds
\[
 \CSEP(\R^n\otimes \R^m)\cap \IPT(\R^n\otimes \R^m)=\SEP(\R^n\otimes \R^m).
\]
 Consequently, if the complexification of an IPT map $\Phi: M_n(\R) \to M_m(\R)$ is $\C$-entanglement breaking, then $\Phi$ is $\R$-entanglement breaking. 
\end{prop}

\begin{proof}
The containment $$\CSEP(\R^n\otimes \R^m)\cap \IPT(\R^n\otimes \R^m)\supseteq\SEP(\R^n\otimes \R^m),$$ is trivial and we noted down this observation before (Lemma \ref{RSEP is IPT}).

For the other containment, If $X\in \CSEP$, then we can write
\[
X = \sum_k P_k \otimes Q_k,
\]
with $P_k, Q_k \ge 0$ complex matrices. Write $P_k = A_k + iB_k, \, Q_k = C_k + i D_k,$ where $A_k,B_k,C_k, D_k \in M_n(\R)$ with $A_k, C_k$ symmetric and $B_k, D_k$ skew-symmetric.  Also note that $A_k, C_k \ge 0$.

Since $X$ is IPT,
\[
X=(\tau_n\otimes \id_m)    (X)  = \sum_k P_k^t \otimes Q_k = \sum_k  (A_k - i B_k)\otimes Q_k
\]
Averaging these two expressions yields $$X= \sum_k  A_k\otimes Q_k=\sum_k (A_k\otimes C_k) +i (A_k\otimes D_k).$$
But since $X$ is a real matrix, we obtain $X=\sum_k A_k\otimes C_k$, which means $X$ is real separable.

The last statement follows from the fact that if $\Phi$ is $\C$-entanglement breaking, then the Choi matrix $C_\Phi$ is in $\CSEP$ (see Theorem \ref{CSEP-breaking}). Now the IPT property of $\Phi$ makes $C_\Phi$ an IPT state. Hence we have $C_\Phi$ to be in $\SEP(\R^n\otimes \R^m)$ and now by Theorem \ref{Rentanglementbreaking characterization} we get $\Phi$ is $\R$-entanglement breaking.
\end{proof}

  Also, the  IPT maps  can be characterized as the completely positive maps that   annihilate the space of antisymmetric matrices: 
\begin{prop}\label{IPTmaps}
For a completely positive map  $\Phi:  M_n (\K)   \to M_m (  \K) $, the following are equivalent 
\begin{enumerate}   
\item[(i)] $\Phi$ is IPT 
\item[(ii)] $\Phi  (A)  =  0 $ for every antisymmetric matrix $A  \in  \Asym_n (\K)$.
\end{enumerate}
\end{prop}

\begin{proof}   (i)  $\Rightarrow$ (ii).  Let $\Phi$ be an IPT map.  Hence, $C_\Phi  =  J_\Phi$.  Hence, for every antisymmetric matrix  $A\in  \Asym_n (\K)$, 
Eq.  (\ref{choiinverse}) implies the relation  
\begin{align*}
\nonumber 
\Tr (\Phi(A)B)     & =    \Tr (C_\Phi (A^t\otimes B)) = -    \Tr (C_\Phi (A\otimes B))    = -  \Tr (  J_\Phi \, (A \otimes B))     \\
  &  =  -  \Tr  (  C_\Phi\,  (A^t \otimes B))   =  -  \Tr (\Phi(A)B)  \, ,  \qquad  \forall B  \in  M_m (\K)   \, ,
\end{align*}
which is equivalent to $\Phi (A)  =  0$.  

  (ii)  $\Rightarrow$ (i).     Let $\Phi$ be a CP map such that $\Phi  (A)  =  0$ for every $A \in \Asym_n  (\K)$.  Hence, for every matrix $X \in  M_ n (\K)$,  we have  $\Phi   (X)   =  \Phi   (X^t )$.    In turn, this condition implies 
  \begin{align}
C_\Phi     =  \sum_{i,j}  E_{ij}    \otimes \Phi (E_{ij} )    =    \sum_{i,j}  E_{ij}    \otimes \Phi (E_{ij}^t )   =  \sum_{i,j}  E_{ij}    \otimes \Phi (E_{ji} )   = J_{\Phi} \, .       \end{align}
Since the Choi matrix and Jamio\l kowski matrices coincide and are positive, $\Phi$ is IPT.  
\end{proof}

\subsection{PPT and IPT squared conjectures}
In a presentation at the workshop  ``Operator structures in quantum information theory''  (Banff Research Station, February 26-March 2 2012) \cite{RJKHW}, Christandl raised what has later  become known as the PPT-squared conjecture.

\begin{conj}[\textbf{PPT-squared}] \label{PPT2C}
For every PPT map $\Phi:M_n(\C)\rightarrow M_n(\C)$, $\Phi\circ\Phi$ is $\C$-entanglement breaking.
\end{conj}

In \cite{CMW}, Christandl, M\"{u}ller-Hermes, and Wolf show that if this conjecture is true, then the composition $\Phi \circ \Psi$ 
is entanglement breaking for any two PPT maps with compatible domain and range.
They also show that the conjecture is true for $n=3$ \cite{CMW}*{Corollary 3.1}.

The example below shows that a direct real-vector-space analogue of Conjecture~\ref{PPT2C} is false. Namely, there exist real PPT maps whose square is not $\R$-entanglement breaking, even for $n=2$.

In \cite{KMP} it is shown that if $\Phi: M_n(\C) \to M_n(\C)$ is idempotent, i.e., $\Phi \circ \Phi = \Phi$, 
unital, i.e, $\Phi(I_n) = I_n$ and PPT, then $\Phi \circ \Phi$ is $\C$-entanglement breaking.  
In \cite{KMP} it is also shown that if $\Phi$ is only unital and PPT, then
every limit point of the sequence of iterates $\Phi^k = \Phi\circ \dots \circ \Phi$ is a $\C$-entanglement breaking map.
In \cite{RJP} it is shown that if $\Phi: M_n(\C) \to M_n(\C)$ is unital, trace-preserving and PPT, 
then $\Phi^k$ is $\C$-entanglement breaking for some power of $\Phi$. Soon after this result, in \cite{HRF}, it was shown that the unitality can be relaxed with the assumption that a trace-preserving map admits a positive invertible fixed point.  

The example below also shows that the direct real-vector-space analogues of these results are false, even for $n=2$.

\begin{example} \label{E:idem PPT not EB}
Define $\Phi : M_2(\R) \to M_2(\R)$ with
\[
  C_\Phi = \frac12 \begin{bmatrix} I_2 & \gamma \\ -\gamma & I_2 \end{bmatrix}
  \quad\text{where}\quad 
  \gamma =\begin{bmatrix} \phantom{-} 0 & 1 \\ -1 & 0 \end{bmatrix}.
\]
Then $\Phi$ is a unital, completely positive, trace preserving map which is idempotent, $\Phi^2=\Phi$.
Moreover, $\ran \Phi = \spn\{ I_2,\gamma \}$ is a real abelian C*-algebra with the usual product.
Since $C_\Phi$ and its partial transpose are positive, $\Phi$ is PPT.
Since $C_\Phi$ is not equal to its partial transpose, Corollary~\ref{C:RPPT equals transpose} shows that $\Phi^2= \Phi$ is not $\R$-entanglement breaking.

However, $C_\Phi$ is in $\CSEP(\R^2\otimes \R^2)$: let $A_\pm = \frac12 (I_2 \pm i\gamma ) \ge 0$. Then
\[
 A_+ \otimes A_- + A_- \otimes A_+ = \frac12(I_2 \otimes I_2 + \gamma \otimes \gamma) = C_\Phi 
\]
and
\[
 A_+ \otimes A_+ + A_- \otimes A_- = \frac12(I_2 \otimes I_2 - \gamma \otimes \gamma) = C_{\Phi \circ T}
\]
Therefore $\wt\Phi^2=\wt\Phi$ is $\C$-entanglement breaking. 
\end{example}

With Corollary~\ref{C:RPPT equals transpose} and Example~\ref{E:idem PPT not EB} in mind,  we formulate a real version of the PPT-squared conjecture, which we call the IPT-squared conjecture.
This conjecture also makes sense for $\C$.
We formulate both versions.

\begin{conj}[\textbf{IPT-squared for $\K$}] \label{IPT2}
For $\K \in \{ \R, \C\}$, every IPT map $\Phi:M_n(\K)\rightarrow M_n(\K)$ satisfies $\Phi\circ\Phi$ is $\K$-entanglement breaking.
\end{conj}

\begin{rem}
It is worth stressing  that not all IPT maps on $M_n(\R)$ are $\R$-entanglement breaking. Indeed, as discussed after Lemma \ref{RSEP is IPT}, Ref. \cite{bennett1999unextendible} provides  examples of unextendible product bases (UPBs) in $\C^3\otimes \C^3$, so that the projector on their orthogonal complement is  a positive matrix $\rho$ in $M_3(\R)\otimes M_3(\R)$ which is IPT but not complex separable and hence not real separable. Now consider the map $\Phi: M_3(\R)\rightarrow M_3(\R)$ whose Choi matrix is $\rho$. It is clear that such a map is IPT but it is not $\R$-entanglement breaking. 
\end{rem}

The following gives a further indication that the IPT-squared conjecture is the correct real counterpart of the original PPT-conjecture.

\begin{prop} 
If the $\PPT^2$ conjecture is true for $\K$, then the $\IPT^2$ Conjecture is true for $\K$. 
If the $\IPT^2$ Conjecture is true for $\C$, then it is true for $\R$.

In particular, in dimension three, the above conjectures are true.
\end{prop}

\begin{proof} 
The first statement is obvious, since Conjecture~\ref{IPT2}  is a special case of Conjecture~\ref{PPT2C}.

Assume that Conjecture~\ref{IPT2} is true for $\C$. Let $\Phi: M_n(\R) \to M_n(\R)$ be IPT. 
Then $\wt\Phi: M_n(\C) \to M_n(\C)$ is IPT and hence $\wt\Phi \circ \wt\Phi = \wt{\Phi \circ \Phi}$ is $\C$-entanglement breaking. This means the Choi matrix $C_{\Phi \circ \Phi} = C_{\wt\Phi \circ \wt\Phi}$ is in $\CSEP(\R^n\otimes \R^m)$. Using Proposition \ref{IPT+CSEP=RSEP} we obtain $C_{\Phi \circ \Phi} $ is real separable and hence $\Phi \circ \Phi$ is $\R$-entanglement breaking.

Note that, it is known that in dimension three (complex case), the PPT-sqaured conjecture is true (see \cite{CMW} and also \cite{CYT}). Hence using the earlier proof method, one can show that the $\IPT^2$ Conjecture is true for $\R$ in dimension 3. 
\end{proof}

We do not know if these conjectures are equivalent.

We can however establish the three results from \cite{KMP} and \cite{RJP} cited above hold for IPT maps.

\begin{prop} 
Let $\Phi: M_n(\R) \to M_n(\R)$ be unital, trace-preserving and IPT. Then there is a power $k$ such that $\Phi^k$ is $\R$-entanglement breaking.
\end{prop}

\begin{proof} 
Since $\wt\Phi$ is unital, trace-preserving and PPT, there is some $k$ so that $(\wt\Phi)^k = \widetilde{\Phi^k}$ is $\C$-entanglement breaking. 
Using the fact that $\tau_n \circ \Phi = \Phi$ implies $\tau_n \circ \Phi^k = \Phi^k$ and arguing as in the previous proof, we obtain that $C_{\Phi^k}$ is $\R$-separable.
\end{proof}

Before proving the analogues of the results in \cite{KMP}, an analogue of \cite{KMP}*{Lemma~3.1} is required. 
However it is much more difficult in the real case.

When $\Phi:M_n(\C) \to M_n(\C)$ is idempotent and unital, i.e., $\Phi\circ\Phi=\Phi$ and $\Phi(I_n) = I_n$, 
then the range of $\Phi$ is an operator system which becomes a C*-algebra when endowed with the Choi-Effros product $A \star B = \Phi(AB)$.
In \cite{KMP}, it is shown that this C*-algebra is abelian.  
They then use the fact that every finite dimensional complex abelian C*-algebra is of the form $\ell^{\infty}_k(\C)$ and the complex analogue of \ref{Rentanglementbreaking characterization}(5), to deduce that $\Phi$ is $\C$-entanglement breaking.

However, not every finite dimensional real abelian C*-algebra is of the form $\ell^{\infty}_k(\R)$.
In particular, in Example~\ref{E:idem PPT not EB}, the range of $\Phi$ is $*$-isomorphic to the real abelian C*-algebra
\[
  \ff A = \{  f \in C(\{a,b\}) : f(b) = \ol{f(a)} \} .
\]
This is a 2-dimensional abelian real C*-algebra which is not isomorphic to $l^\infty_2(\R)$ or $l^\infty_1(\C)$.
Note that $\Phi$ is PPT but not IPT.

\begin{thm}
Let $\Psi : M_n(\R) \to M_n(\R)$ be a unital IPT idempotent map. 
Then the operator system $\cl R := \Psi(M_n(\R))$ with the Choi-Effros product
is order isomorphic to $l^\infty_k(\R)$ for some $k\ge1$, and $\Psi$  is $\R$-entanglement breaking.
\end{thm}

\begin{proof}
The Choi-Effros product on $\cl R$ is given by $A \star B := \Psi(AB)$,
and this makes $\cl R$ into a real C*-algebra $\ff A$.
The complexification $\wt\Psi$ is PPT, unital and idempotent.
By \cite{KMP}*{Lemma~3.1}, the Choi-Effros product on $\wt\Psi(M_n(\C))$ is abelian;
and hence $\ff A$ is abelian.
Let $\Gamma : \cl R \to \ff A$ be this identification, which is a complete isometry; so $\Gamma^{-1}$ is also completely isometric.

An old result of Arens and Kaplansky shows that if $\ff A$ is an abelian real C*-algebra,
then there is a locally compact Hausdorff space $Z$ and a homeomorphism $\tau:Z \to Z$ with $\tau^2=\id_Z$
such that 
\[
 \ff A \simeq \{ f \in C_0(Z) : f(\tau(z)) = \ol{f(z)} \FORAL x \in Z \} .
\]
See \cite{rosenberg}*{Theorem 1.9} for an easy proof.
In our case, since $\ff A$ is finite dimensional, $X$ is a finite set with the discrete topology.
Hence 
\[
  Z = \{ x_j, y_k, z_k: 1 \le j \le J, 1 \le k \le K \}
\]
such that $\tau(x_j)=x_j$, $\tau(y_k) = z_k$ and $\tau(z_k) = y_k$. Thus
\[
  \ff A = \Big\{ \sum_j a_j \delta_{x_j} + \sum_k b_k \delta_{y_k} + \ol{b_k} \delta_{z_k}: a_j \in \bb R, \, b_k \in \bb C \Big\}. 
\]
A basis for $\ff A$ as a real vector space is given by $\{ \delta_{x_j}, \delta_{y_k} + \delta_{z_k}, i(\delta_{y_k} - \delta_{z_k}) \}$. 
Write these vectors as $\{ \alpha_j, \beta_k, \gamma_k \}$. 
Then $\alpha_j^*= \alpha_j$, $\beta_k^* = \beta_k$ and $\gamma_k^* = - \gamma_k$.
 
The map $\Phi = \Gamma \Psi : M_n(\bb R) \to \ff A$ can be written as
\[
  \Phi(X) = \sum_j \Tr(XA_j) \alpha_j + \sum_k \Tr(XB_k) \beta_k + \Tr(XG_k) \gamma_k
\]
for matrices $A_j, B_k, G_k \in M_n(\bb R)$.
Since $\Phi$ is CP, 
\[
 X \ge 0 \implies \Tr(XA_j) \ge 0,\ \Tr(X B_k) \ge 0 \AND \Tr(XG_k) =0 .
\]
Therefore $A_j, B_k \in \PSD(\R^n)$ and $G_k^t = -G_k$, since $G_k$ is orthogonal to all positive matrices, and hence to all symmetric matrices.

Let $P_j = \Gamma^{-1}(\alpha_j)$, $Q_k = \Gamma^{-1}(\beta_k)$ and $R_k = \Gamma^{-1}(\gamma_k)$.
Then $P_j, Q_k$ are positive, and $R_k^t = - R_k$ because $\Gamma$ and $\Gamma^{-1}$ are positive.
This is a real basis for $\cl R$.
Now we have
\[
 \Psi(X) = \Gamma^{-1} \Phi(X) = \sum_j \Tr(XA_j) P_j + \sum_k \Tr(XB_k) Q_k + \Tr(XG_k) R_k . 
\]
The Choi matrix is
\[
 C_\Psi = \sum_j A_j \otimes P_j + \sum_k B_k \otimes Q_k + \sum_k G_k \otimes R_k.
\]
The partial transpose of this map is
\[
 C_{T \circ\Psi} = \sum_j A_j \otimes P_j + \sum_k B_k \otimes Q_k - \sum_k G_k \otimes R_k.
\]
However, $C_{T \circ\Psi} = C_\Psi$ since $\Psi$ is IPT. So $G_k = 0$.
Therefore $C_\Psi$ is separable, and hence $\Psi$ is $\R$-entanglement breaking by Theorem~\ref{Rentanglementbreaking characterization}.

Now $\Gamma$ is an isomorphism, and in particular it is surjective.
Therefore $K=0$ and $Z = \{ x_j : 1 \le j \le J \}$.
Hence $\ff A \simeq C_\R(Z) \simeq l^\infty_J(\R)$.
\end{proof}

Now the asymptotic result follows as in \cite{KMP} by applying Ellis's Theorem that every compact semigroup contains idempotents (and its proof).

\begin{thm} \label{T:asymptotic REB}
Let $\Phi:M_n(\R)\rightarrow M_n(\R)$ be a unital IPT  map.
Then 
\[
 \lim_{k\to\infty} \dist(\Phi^k, \R\text{-}\EB  (n,n)) = 0 . 
\]
\end{thm}
\bigskip

\section*{Acknowledgments}  

GC  was supported by the Hong Kong Research Grant Council through grant 17300920  and through 
the Senior Research Fellowship Scheme via SRFS2021-7S02, and by the Croucher foundation. 
KRD and VIP were partially supported by the Natural Sciences and Engineering Research Council of Canada (NSERC).
MR is supported by the European Research Council (ERC Grant Agreement No. 851716).

The authors would like to thank the anonymous referees for many helpful suggestions in this article.

\bibliographystyle{amsplain}

\end{document}